\PassOptionsToPackage{unicode}{hyperref}
\PassOptionsToPackage{hyphens}{url}
\PassOptionsToPackage{dvipsnames,svgnames,x11names}{xcolor}
\documentclass[
  12pt]{article}
\usepackage{longtable}
\usepackage{bm}
\usepackage{booktabs}
\usepackage{multirow}
\usepackage{amsmath,amssymb}
\usepackage{subcaption}
\usepackage{amsthm}
\usepackage{iftex}
\usepackage{float}
\ifPDFTeX
  \usepackage[T1]{fontenc}
  \usepackage[utf8]{inputenc}
  \usepackage{textcomp} 
\else 
  \usepackage{unicode-math}
  \defaultfontfeatures{Scale=MatchLowercase}
  \defaultfontfeatures[\rmfamily]{Ligatures=TeX,Scale=1}
\fi
\usepackage{lmodern}
\ifPDFTeX\else  
\fi
\IfFileExists{upquote.sty}{\usepackage{upquote}}{}
\IfFileExists{microtype.sty}{
  \usepackage[]{microtype}
  \UseMicrotypeSet[protrusion]{basicmath} 
}{}
\makeatletter
\@ifundefined{KOMAClassName}{
  \IfFileExists{parskip.sty}{%
    \usepackage{parskip}
  }{
    \setlength{\parindent}{0pt}
    \setlength{\parskip}{6pt plus 2pt minus 1pt}}
}{
  \KOMAoptions{parskip=half}}
\makeatother
\usepackage{xcolor}
\makeatletter
\ifx\paragraph\undefined\else
  \let\oldparagraph\paragraph
  \renewcommand{\paragraph}{
    \@ifstar
      \xxxParagraphStar
      \xxxParagraphNoStar
  }
  \newcommand{\xxxParagraphStar}[1]{\oldparagraph*{#1}\mbox{}}
  \newcommand{\xxxParagraphNoStar}[1]{\oldparagraph{#1}\mbox{}}
\fi
\ifx\subparagraph\undefined\else
  \let\oldsubparagraph\subparagraph
  \renewcommand{\subparagraph}{
    \@ifstar
      \xxxSubParagraphStar
      \xxxSubParagraphNoStar
  }
  \newcommand{\xxxSubParagraphStar}[1]{\oldsubparagraph*{#1}\mbox{}}
  \newcommand{\xxxSubParagraphNoStar}[1]{\oldsubparagraph{#1}\mbox{}}
\fi
\makeatother

\usepackage{longtable,booktabs,array}
\usepackage{calc} 
\usepackage{etoolbox}
\makeatletter
\patchcmd\longtable{\par}{\if@noskipsec\mbox{}\fi\par}{}{}
\makeatother
\IfFileExists{footnotehyper.sty}{\usepackage{footnotehyper}}{\usepackage{footnote}}
\makesavenoteenv{longtable}
\usepackage{graphicx}
\makeatletter
\def\maxwidth{\ifdim\Gin@nat@width>\linewidth\linewidth\else\Gin@nat@width\fi}
\def\maxheight{\ifdim\Gin@nat@height>\textheight\textheight\else\Gin@nat@height\fi}
\makeatother
\setkeys{Gin}{width=\maxwidth,height=\maxheight,keepaspectratio}
\makeatletter
\def\fps@figure{htbp}
\makeatother

\makeatletter
\@ifpackageloaded{caption}{}{\usepackage{caption}}
\AtBeginDocument{%
\ifdefined\contentsname
  \renewcommand*\contentsname{Table of contents}
\else
  \newcommand\contentsname{Table of contents}
\fi
\ifdefined\listfigurename
  \renewcommand*\listfigurename{List of Figures}
\else
  \newcommand\listfigurename{List of Figures}
\fi
\ifdefined\listtablename
  \renewcommand*\listtablename{List of Tables}
\else
  \newcommand\listtablename{List of Tables}
\fi
\ifdefined\figurename
  \renewcommand*\figurename{Figure}
\else
  \newcommand\figurename{Figure}
\fi
\ifdefined\tablename
  \renewcommand*\tablename{Table}
\else
  \newcommand\tablename{Table}
\fi
}
\@ifpackageloaded{float}{}{\usepackage{float}}
\floatstyle{ruled}
\@ifundefined{c@chapter}{\newfloat{codelisting}{h}{lop}}{\newfloat{codelisting}{h}{lop}[chapter]}
\floatname{codelisting}{Listing}

\makeatother
\makeatletter
\@ifpackageloaded{caption}{}{\usepackage{caption}}
\@ifpackageloaded{subcaption}{}{\usepackage{subcaption}}
\makeatother

\ifLuaTeX
  \usepackage{selnolig}  
\fi
\usepackage[]{natbib}
\usepackage{bookmark}

\IfFileExists{xurl.sty}{\usepackage{xurl}}{} 
\hypersetup{
  pdftitle={A Bayesian Framework for Quantifying Association Between Functional and Structural Data in Neuroimaging},
  pdfauthor={Sakul Mahat; Sharmistha Guha; Jessica Bernard},
  pdfkeywords={High dimensional data, Bayesian Statistics, Brain connectivity, Network-Structural association, Variable selection, Latent factor models},
  colorlinks=true,
  linkcolor={blue},
  filecolor={Maroon},
  citecolor={Blue},
  urlcolor={Blue},
  pdfcreator={LaTeX via pandoc}}

\newcommand{\anon}{1}

\begin{document}

\def\spacingset#1{\renewcommand{\baselinestretch}%
{#1}\small\normalsize} \spacingset{1}


\if1\anon
{
  \title{\bf A Bayesian Framework for Quantifying Association Between Functional and Structural Data in Neuroimaging}
  \author{Sakul Mahat \thanks{
    The authors gratefully acknowledge \textit{the advanced computing resources provided by the Texas A\&M Department of Statistics Arseven Computing Cluster.}}\hspace{.2cm}\\
    Department of Statistics, Texas A\&M University\\
    and \\
    Sharmistha Guha \\
    Department of Statistics, Texas A\&M University\\
    and \\
    Jessica Bernard \\
    Department of Psychological \& Brain Sciences, Texas A\&M University}
  \maketitle
} \fi

\if0\anon
{
  \bigskip
  \bigskip
  \bigskip
  \begin{center}
    {\LARGE\bf Title}
\end{center}
  \medskip
} \fi



\bigskip
\begin{abstract}
Structural and functional neuroimaging modalities provide distinct yet complementary windows into brain organization. Structural imaging characterizes the anatomy and microstructure of neural tissue, elucidating the physical architecture that supports brain function. Functional imaging, in contrast, captures dynamic patterns of neural activity and connectivity, revealing how information is processed and transmitted across brain regions in real time. Together, these modalities offer a more complete picture of brain organization than either could provide alone. Recent advances in multimodal neuroimaging have led to a growing body of work on joint modeling of structural and functional data, with the goal of deepening our understanding of the neurological processes that arise from their interplay. Much of this literature operates under the assumption of a strong association between brain structure and function, and exploits this relationship to improve prediction and interpretability. However, relatively little attention has been given to developing rigorous, statistically principled frameworks for formally testing hypotheses about the strength and nature of these associations. Existing approaches typically rely on simple correlation-based measures or heuristic integration strategies, which may fail to capture the complex dependencies and structured variability inherent in neuroimaging data, particularly when functional data are represented as brain networks and structural data as region-specific anatomical measures.
In this work, we address this gap by developing an explicit Bayesian hypothesis testing framework for quantifying associations between structural and functional neuroimaging data. Our approach first constructs functional brain networks from fMRI data, then integrates them with structural measurements through a hierarchical Bayesian model derived from complementary modalities. The Bayesian formulation provides a coherent probabilistic framework that naturally accommodates two types of datasets with different structures, incorporates prior knowledge, and yields full posterior uncertainty quantification. Through extensive empirical studies, we demonstrate that the proposed method achieves excellent performance in detecting associations under a wide range of settings, including varying signal-to-noise ratios, different numbers of brain regions, and diverse sets of structural imaging measures. 
\end{abstract}


\noindent%
\vfill

\newpage
\spacingset{1.8} 



\section{Introduction}\label{sec-intro}
Structural and functional neuroimaging provide complementary perspectives on the organization and operation of the human brain. Structural modalities such as structural magnetic resonance imaging (sMRI) and diffusion tensor imaging (DTI) characterize gray matter morphology, cortical thickness, white matter integrity, and microstructural properties, yielding detailed maps of the brain’s physical architecture \citep{ashburner2000voxel}. Functional modalities, most notably functional MRI (fMRI), capture time-varying patterns of blood-oxygen-level-dependent (BOLD) signal fluctuations, which are used to infer neural activation and inter-regional coupling at rest or during tasks \citep{ogawa1990brain, smith2013functional}. Together, these modalities offer a more holistic view of brain organization: structural measures describe the anatomical scaffold, whereas functional data reveal the dynamic activity and communication/connectivity that this scaffold supports \citep{park2013structural}.

Over the last decade, there has been substantial interest in jointly modeling structural and functional neuroimaging data to investigate how brain architecture shapes functional dynamics. Correlation-based and network-analytic studies have shown that structural connectivity constrains functional interactions at both global and regional scales, with white matter pathways supporting and limiting patterns of functional connectivity \citep{honey2010can,hermundstad2013structural, honey2009predicting}. Multivariate statistical techniques, including multivariate regression, canonical correlation analysis, and related fusion methods, have been employed to identify shared components and joint patterns of variability across modalities \citep{sui2012review, calhoun2016multimodal}. Graph-theoretical approaches have further enabled comparison of structural and functional brain networks in terms of modularity, efficiency, and hub structure \citep{bullmore2009complex}. More recently, machine learning frameworks, including multi-kernel learning and deep neural network models, have been proposed to leverage cross-modal features for improved prediction of clinical or cognitive outcomes \citep{kawahara2017brainnetcnn, parisot2018disease}.

Within supervised learning, multimodal integration has been explored in scalar-on-image, image-on-scalar, and image-on-image regression frameworks, all aimed at fusing information across imaging modalities to address complex biomedical questions. Image-on-scalar regression treats imaging data (from multiple modalities) as responses and develops joint models that relate these outcomes to subject-level clinical and biological predictors. Recent work has, for example, jointly regressed structural and functional measures on brain-related phenotypes to identify regions associated with cognitive disorders, neurodegeneration, or aging trajectories \citep{gutierrez2025multiobject, rodriguez2025bayesian}.

In contrast, most supervised frameworks treat multimodal imaging as predictors rather than responses, leading to scalar-on-image regression settings. For example, \citet{xue2018bayesian, gutierrez2024regression, guha2024bayesian} develop methods for scalar-on-multimodal-predictor regression, allowing nonlinear relationships between a scalar response and high-dimensional imaging predictors. Image-on-image regression has also been considered to develop predictive models of one imaging modality based on others; for example, \citet{guo2022spatial, ma2023bayesian} propose Bayesian regression models that predict a single response image from multiple image predictors, providing a principled way to interpolate or impute imaging modalities.

Parallel to these statistical approaches, generative deep learning methods have attracted increasing attention for multimodal neuroimaging. Such models seek to synthesize structural or functional images conditional on covariates (e.g., diagnosis, age) or other imaging modalities, using architectures such as conditional GANs, variational autoencoders and Monte Carlo dropout architecture \citep{suzuki2017overview, lei2024inva, jeon2025deep, coombes2015weighted, bowles2018gan, wolterink2017deep}, leading to tools offering powerful predictive capabilities.

The overarching philosophy behind multimodal data integration is that shared information across modalities can be exploited to improve inference, prediction, or both. However, most existing frameworks are designed primarily for predictive modeling or for joint representation learning under an implicit assumption of association, rather than for formal inference on the strength and nature of structural–functional relationships. As a result, conclusions about structure–function coupling are frequently based on heuristic thresholds (e.g., magnitude of loadings, empirical correlations) rather than on statistically principled probabilistic measures of association.

Several studies have attempted to quantify structure–function relationships using simple correlations, similarity metrics, or low-rank linear projections \citep{honey2010can, hermundstad2013structural, mivsic2016network}. While informative, these approaches can be limited in their ability to capture nonlinear and higher-order dependencies, especially when functional data are graph-valued (network connectivity) and structural data are heterogeneous (e.g., a mix of volumetric, cortical thickness, and diffusion-based metrics). Moreover, many existing methods for testing cross-modal association vectorize each modality \citep{sui2009canonical, correa2010canonical}, thereby ignoring the inherent spatial, network, or topological structure of the data and potentially reducing power or interpretability.

Motivated by these gaps, we develop a fully Bayesian framework to quantify the strength of association between structural and functional  data. The central idea is to represent functional brain connectivity as a graph whose nodes correspond to anatomically defined regions of interest (ROIs) and whose edges capture pairwise correlations in fMRI time series between ROIs. This graph representation preserves the topological structure of functional interactions while allowing integration with region-specific structural variables, such as ROI-level measures extracted from T1- and T2-weighted structural MRI scans. Our model introduces a shared set of latent vectors that simultaneously capture two key sources of dependence: (a) the transitivity and community structure inherent in the functional connectivity network, reflecting the tendency of ROIs with strong mutual connections to also connect to common neighbors; (b) the covariation among structural measures across ROIs, enabling joint modeling of anatomy and connectivity within a unified latent space. By coupling the latent representation of structural and functional modalities, our approach provides a coherent Bayesian framework for testing associations at multiple levels. Specifically, it enables (i) ROI-level inference, where each structural variable can be tested for association with the corresponding node’s connectivity profile in the functional graph, and (ii) global inference, which assesses the overall strength of association between the entire set of structural and functional measurements. The Bayesian formulation not only facilitates principled testing but also yields full posterior uncertainty quantification, ensuring interpretable inference.

\section{Multi-Modal Imaging Data To Study Brain Functional \& Structural Association}\label{sec-dataintro}

\noindent\underline{\textbf{Data Collection:}} This article focuses on a multi-modal human neuroimaging data from the \textbf{Lifespan Cognitive and Motor Neuroimaging Laboratory} at Texas A\&M University, College Station. The data set, collected first-hand by co-author Bernard, originates from a longitudinal study of 138 healthy adults (age 35-86, mean age 57, 54\% female) tracking changes in brain and behavior with age. The final analytical cohort included 126 participants after removing individuals with incomplete or low-quality data. For each participant, the dataset comprises functional and structural/anatomical neuroimaging scans, together with subject-specific clinical and behavioral assessments. Eligibility for this study was contingent on the absence of specific medical conditions. Subjects were deemed ineligible if they had a documented history of neurological or psychiatric disorders (including stroke, depression, or anxiety). Exclusion was also based on contraindications to the magnetic resonance imaging environment, such as metallic implants or an inability to tolerate the procedure. 

Participants were scanned using a 3.0 Tesla Siemens Magnetom Verio scanner fitted with a 32-channel head coil. The collection of high-resolution \textbf{structural images} involved a 7-minute T1-weighted MPRAGE sequence and a 5.5-minute T2-weighted sequence, both providing 0.8 mm\textsuperscript{3} isotropic voxels and accelerated with a multiband factor of 2 (TRs of 2400 ms and 3200 ms, respectively). Following scans to obtain brain structural images, resting-state brain \textbf{functional} activity was measured through four separate blood oxygen level dependent (BOLD) fMRI runs, each using a rapid sampling rate (TR = 720 ms) and high acceleration (multiband = 8) to capture 488 brain volumes with 2.5 mm³ voxels. The imaging parameters were adapted from established standards set by the Human Connectome Project (\cite{Harms_2018}) and the University of Minnesota's Center for Magnetic Resonance Research to ensure alignment with best practices for data sharing and reproducibility. The preprocessing pipeline mirrored the approach established in recent work on this dataset by \cite{Hicks_2023} and \cite{Ballard_2022}. Initially, the raw DICOM images were transformed into the NIFTI format and structured according to the Brain Imaging Data Structure (BIDS, version 1.6.0) using the \texttt{bidskit} docker container. To enable distortion correction, B0 field maps were generated by extracting a single volume from two BOLD images with opposing phase-encoding directions, a task performed with the FSL split tool (\cite{Jenkinson_2012}). Finally, the complete dataset underwent comprehensive anatomical and functional preprocessing via fMRIPrep (\url{https://fmriprep.org/}), an automated workflow that prepares neuroimaging data for statistical analysis.

\noindent\underline{\textbf{Brain Functional and Structural Data:}}
The regions of interest (ROIs) were defined according to six established large-scale brain networks following \cite{Jackson_2023}: the default mode, frontoparietal, control, emotion, motor, and a subcortical cerebellar–basal ganglia network. From these networks, a set of 69 ROIs was delineated to ensure that the functional data were spatially aligned with the corresponding structural segmentations. For each participant, BOLD fMRI time series were extracted from these 69 ROIs (illustrated in Figure~\ref{fig:General_workflow_paper}), yielding the functional neuroimaging measurements. Structural information on the volume of each ROI was derived from T1- and T2-weighted sequences, which were likewise summarized at each ROI. The functional Montreal Neurological Institute (MNI) coordinates reported in \cite{Jackson_2023} were anatomically labeled using MRIcroGL by mapping them to regions in the Automated Anatomical Labeling (AAL) atlas \citep{TZOURIOMAZOYER2002273}.

\noindent\underline{\textbf{Subject-Level Measures:}}
 To quantify fine motor function, participants completed the Purdue Pegboard task \citep{lawson2019purdue}. This test involved manipulating pegs, washers, and cylinders on a pegboard under four conditions, each repeated four times. The conditions included: (1) placing pegs with the right hand, (2) with the left hand, (3) with both hands simultaneously, and (4) an assembly task requiring alternating hands. Performance was scored by averaging the number of items completed in each condition. A total score across all conditions was then calculated for each participant to create a single behavioral measure, referred to as the \textbf{Aggregate Pegboard Score (APS)}. 

 \noindent \underline{\textbf{Aims of the Data Analysis.}} The primary objective of this article is to develop and apply a statistical framework that evaluates the overall association between structural and functional imaging modalities, while adjusting for a behavioral predictor (APS). As a structural measure, we use the volume for each ROI, and functional data are used to construct a brain network over ROIs. Within this framework, we seek to quantify the strength of structure-function associations in a probabilistic fashion, that is, to estimate the posterior probability of association between the two modalities after accounting for subject-level covariates. This probability-of-association formulation yields a direct and interpretable measure of evidence for structure-function coupling at both global and feature-specific scales.
 
\section{Model Formulation}\label{sec-modelformulation}

\subsection{Notations}
Consider a brain parcellation into $V$ regions of interest (ROIs) $\mathcal{R}_1,...,\mathcal{R}_V$ defined according to a given brain atlas. For subject $i=1,...,N$, let $\{f_{i,1}(t),...,f_{i,V}(t)\}$, $t=t_1,...,t_T$, denote the BOLD fMRI time series recorded from each ROI. We construct the functional connectivity graph for subject $i$ denoted $\mathbf{A}_{i}\in\mathbb{R}^{V\times V}$, where the $V$ nodes correspond to the ROIs. The edge weight $a_{i,v,v'}$ quantifies the strength of association between regions $\mathcal{R}_v$ and $\mathcal{R}_{v'}$ and is given by the Fisher $z$-transformed Pearson correlation between the time series $(f_{i,v}(t_1),...,f_{i,v}(t_T))^T$ and $(f_{i,v'}(t_1),...,f_{i,v'}(t_T))^T$. By construction, the connectivity graph is undirected and has no self-loops, so \( a_{i,v,v'} = a_{i,v',v} \) and \( a_{i,v,v} = 0 \). Similarly, let \( \mathbf{Y}_i \in \mathbb{R}^{V \times P} \) represent the structural data for subject \( i \), where $y_{i,v,k}$ is the measurement of $k$th structural variable ($k=1,...,P$) at ROI $\mathcal{R}_v$. Thus, for each subject, we observe two complementary data sources, the functional connectivity matrix $\mathbf{A}_i$ encoding pairwise associations between ROIs; and the structural feature matrix $\mathbf{Y}_i$, containing multiple anatomical or microstructural measures for each ROI. We also record a covariate $x_i \in \mathbb{R}$ representing a subject-specific behavioral characteristic, namely the APS.

This section presents the proposed modeling framework, including the latent space formulation for graph connectivity and structural variables, the prior structure, the inferential quantities used to assess associations, and details of posterior computation.
\subsection{Latent Space Model for Brain Connectivity Graph and Structural Variables}
We adopt a conditionally random effects model with shared latent effects to characterize the association between brain connectivity graphs and structural variables. The key idea is to introduce ROI-specific latent factors that drive both brain connectivity and structural variation, thereby creating a principled mechanism for linking the two modalities and testing association.

\noindent\textbf{Brain Connectivity Graph Model.}
The observed connectivity $a_{i,v,v'}$ between regions \( \mathcal{R}_v \) and \( \mathcal{R}_{v'} \) for the $i$th subject is modeled as a function of ROI-specific latent variables:
\begin{equation}
    a_{i,v,v'} = \mu_i + x_i \beta_{x,G} + h(\boldsymbol{\eta}_{i,v}, \boldsymbol{\eta}_{i,v'}, \boldsymbol{\lambda}) + \epsilon_{i,v,v'}, 
    \qquad \epsilon_{i,v,v'} \sim N(0, \sigma^2),
    \label{eq:network_model}
\end{equation}
where $\mu_i$ is the subject-specific intercept, $\beta_{x,G} \in \mathbb{R}$ denotes the coefficient for the subject-specific behavioral covariate (APS), $\boldsymbol{\eta}_{i,v} = (\eta_{i,v,1}, \dots, \eta_{i,v,R})^T$ and $\boldsymbol{\eta}_{i,v'} = (\eta_{i,v',1}, \dots, \eta_{i,v',R})^T$ are the $R$-dimensional latent factors corresponding to the ROIs $\mathcal{R}_v$ and $\mathcal{R}_{v'}$, respectively, $\boldsymbol{\lambda} = (\lambda_1, \dots, \lambda_R)^T$ is a vector of inclusion variables, with $\lambda_r = 0$ signifying exclusion of the $r$-th latent dimension from explaining connectivity, $h(\cdot)$ encodes the functional form through which ROI-level latent factors and inclusion indicators shape the edge strengths, and $\sigma^2$ is the variance of idiosyncratic errors.

\noindent \textbf{Structural Variables Model.} The structural data are modeled in a way that shares ROI-specific latent factors with the brain connectivity model, thus enabling direct assessment of structural–functional associations. For subject $i$, the structural measurement of variable $k$ at ROI $\mathcal{R}_v$, denoted by $y_{i,v,k}$, is given by,
\begin{equation}
    y_{i,v,k} = \theta_{i,k} + x_i \beta_{x,S} + g(\boldsymbol{\eta}_{i,v}, \boldsymbol{\alpha}_{i,k}, \boldsymbol{\omega}_k) + \delta_{i,v,k}, 
    \qquad \delta_{i,v,k} \sim N(0, \phi_k^2), 
    \label{eq:structural_model}
\end{equation}

where $\theta_{i,k}$ is the subject- and variable-specific intercept, 
$\beta_{x,S} \in \mathbb{R}$ denotes the coefficient for the subject-specific behavioral covariate (APS), 
$\boldsymbol{\alpha}_{i,k} = (\alpha_{i,k,1}, \dots, \alpha_{i,k,R})^T$ is a structural variable-specific latent factor for each subject, 
and $\boldsymbol{\omega}_k = (\omega_{k,1}, \dots, \omega_{k,R})^T$ are the inclusion variables for the dimensions of latent factors. 
Here, $g(\cdot)$ defines how ROI-level latent factors, variable-specific latent factors, and inclusion indicators combine to explain variability in $y_{i,v,k}$, 
and $\phi_k^2$ is the variance of idiosyncratic errors. 
By sharing the ROI-specific latent factors $\boldsymbol{\eta}_{i,v}$ across the two models, the framework directly links structural variables with functional connectivity, 
thereby facilitating both ROI-level and global association tests across modalities, as elaborated in the following sections.

\noindent\textbf{Choice of $h(\cdot)$ and $g(\cdot)$.} To characterize the relationship between latent factors and the observed data, we adopt a bilinear interaction structure \citep{athreya2018statistical, guha2021bayesian} for both functions. Specifically,

\begin{equation}\label{eq:bil}
h(\boldsymbol{\eta}_{i,v}, \boldsymbol{\eta}_{i,v'}, \boldsymbol{\lambda}) = \sum_{r=1}^R \eta_{i,v,r}\eta_{i,v',r}\lambda_r, \quad g(\boldsymbol{\eta}_{i,v}, \boldsymbol{\alpha}_{i,k}, \boldsymbol{\omega}_k) = \sum_{r=1}^R \eta_{i,v,r}\alpha_{i,k,r}\omega_{k,r},
\end{equation}

In the graph model, the bilinear construction of $h(\cdot)$ is well-suited to capture important topological features of graphs such as transitivity, clustering, and community organization, which are hallmark properties of brain connectivity graphs. In the structural model, $g(\cdot)$ allows structural variables to be explained through the alignment of variable-specific latent factors with the same latent dimensions that govern connectivity. This enables the framework to jointly represent how anatomical measurements covary with graph structure across subjects. The bilinear structure in (\ref{eq:bil}) naturally ensures that if $\lambda_r=0$, $r$-th latent dimension does not contribute to graph connectivity, while if $\omega_{k,r}=0$ the same dimension does not explain variability in the $k$-th structural variable.


\subsection{Hierarchical Prior Specification}
To complete the model formulation, we specify priors for latent factors, regression coefficients, and inclusion variables. All priors are chosen to balance flexibility, sparsity, and computational tractability.

All latent factors, baseline intercepts, and coefficients associated with demographic and clinical covariates are assigned Gaussian priors, thereby allowing the observed data to drive inference on subject- and region-specific effects. In particular, the intercept for the graph model $\mu_i$, and the structural variable-specific intercepts $\theta_{i,k}$ for each subject are given $N(0,1)$ priors. ROI-specific latent factors $\boldsymbol{\eta}_{i,v}$, as well as structural variable-specific latent factors $\boldsymbol{\alpha}_{i,k}$ are assigned $N(\mathbf{0}, \sigma_\eta^2 \mathbf{I}_R)$ and $N(\mathbf{0}, \sigma_\alpha^2 \mathbf{I}_R)$ priors, respectively. The coefficients for the behavioral covariate (APS), $\beta_{x,G}$ and $\beta_{x,S}$, each follow a standard normal prior.

A central feature of the model is the ability to automatically select relevant latent dimensions that explain structural–functional associations. To this end, we place spike-and-slab priors on the inclusion variables.

\textbf{Prior on graph model inclusion variable.} Each interaction parameter $\lambda_r$ is expressed as
\begin{equation}
 \lambda_r \sim \gamma_r N(0, \tau_r^{-1}) + (1 - \gamma_r) \delta_0,    
\end{equation}
where $\gamma_r \in \{0, 1\}$ is the inclusion indicator for the $r$-th latent dimension. The prior distribution $\gamma_r \sim \text{Bernoulli}(\pi_\lambda)$ determines whether a dimension contributes to connectivity. Precision parameters are drawn from $\tau_r \sim \text{Gamma}(q^{3(r-1)}, q^{2(r-1)})$, which encourages shrinkage of higher-order latent dimensions.

\textbf{Prior on structural model inclusion variable.} Analogously, each structural variable activates a subset of latent dimensions through
\begin{equation}
\omega_{k,r} \sim \psi_{k,r} N(0, \nu_{k,r}^{-1}) + (1 - \psi_{k,r}) \delta_0,   
\end{equation}
with inclusion indicators $\psi_{k,r} \sim \text{Bernoulli}(\pi_{\psi,k})$. The precision terms follow $\nu_{k,r} \sim \text{Gamma}(q^{3(r-1)}, q^{2(r-1)})$, allowing structural modalities to selectively engage relevant latent dimensions, creating flexibility across variables and subjects.

\textbf{Hyperpriors for mixing probabilities and variance parameters.}
The inclusion probabilities are given Beta hyperpriors, $\pi_\lambda \sim \text{Beta}(a_\lambda, b_\lambda)$ and $\pi_{\psi,k} \sim \text{Beta}(a_\psi, b_\psi)$, which adaptively regulate sparsity both at the latent dimension level and across structural modalities. All variance components, $\sigma^2$, $\phi_k^2$, $\sigma_\eta^2, \sigma_\alpha^2$, are assigned conjugate Inverse-Gamma priors, $\text{IG}(a,b)$.


\subsection{Global and Local Association between Brain Connectivity Graph and Structural Variables}

A central scientific objective is to quantify the extent of shared latent structure between graph connectivity and structural modalities. To this end, we consider two levels of testing: (a) \emph{Local association testing} evaluates the degree of association between each individual structural variable and the graph, thereby identifying structural features that share latent dimensions with functional connectivity; (b) \emph{Global association testing} evaluates the joint association across all structural variables and the graph, thereby quantifying the overall strength of structural-functional dependence.

\textbf{Local association testing.} Let $\boldsymbol{\gamma} = (\gamma_1, \dots, \gamma_R)^T$ denote the $R$-dimensional vector of inclusion indicators corresponding to latent dimensions in the graph model, and let $\boldsymbol{\psi}_k = (\psi_{k,1}, \dots, \psi_{k,R})^T$ denote the corresponding inclusion indicators for the $k$-th structural variable. We define the local association metric for the structural variable $k$ as
\begin{equation}
A_{L,k}(\boldsymbol{\gamma}, \boldsymbol{\psi}_k) = \boldsymbol{\gamma}^T \boldsymbol{\psi}_k.
\end{equation}
This simple inner product measures the number of latent dimensions jointly active in both the graph model ($\gamma_r = 1$) and the $k$-th structural variable ($\psi_{k,r} = 1$). Hence, a larger value of the posterior probability $P(A_{L,k}(\boldsymbol{\gamma}, \boldsymbol{\psi}_k) > 0 | \text{data})$ indicates that the $k$-th structural variable shares more latent structure with the functional graph, while $P(A_{L,k}(\boldsymbol{\gamma}, \boldsymbol{\psi}_k) > 0 | \text{data}) = 0$ implies that no dimensions are shared. Local testing thus pinpoints which structural variables drive functional connectivity through overlapping latent representations.

\begin{figure}[H]
    \centering
    \includegraphics[width=0.9\textwidth]{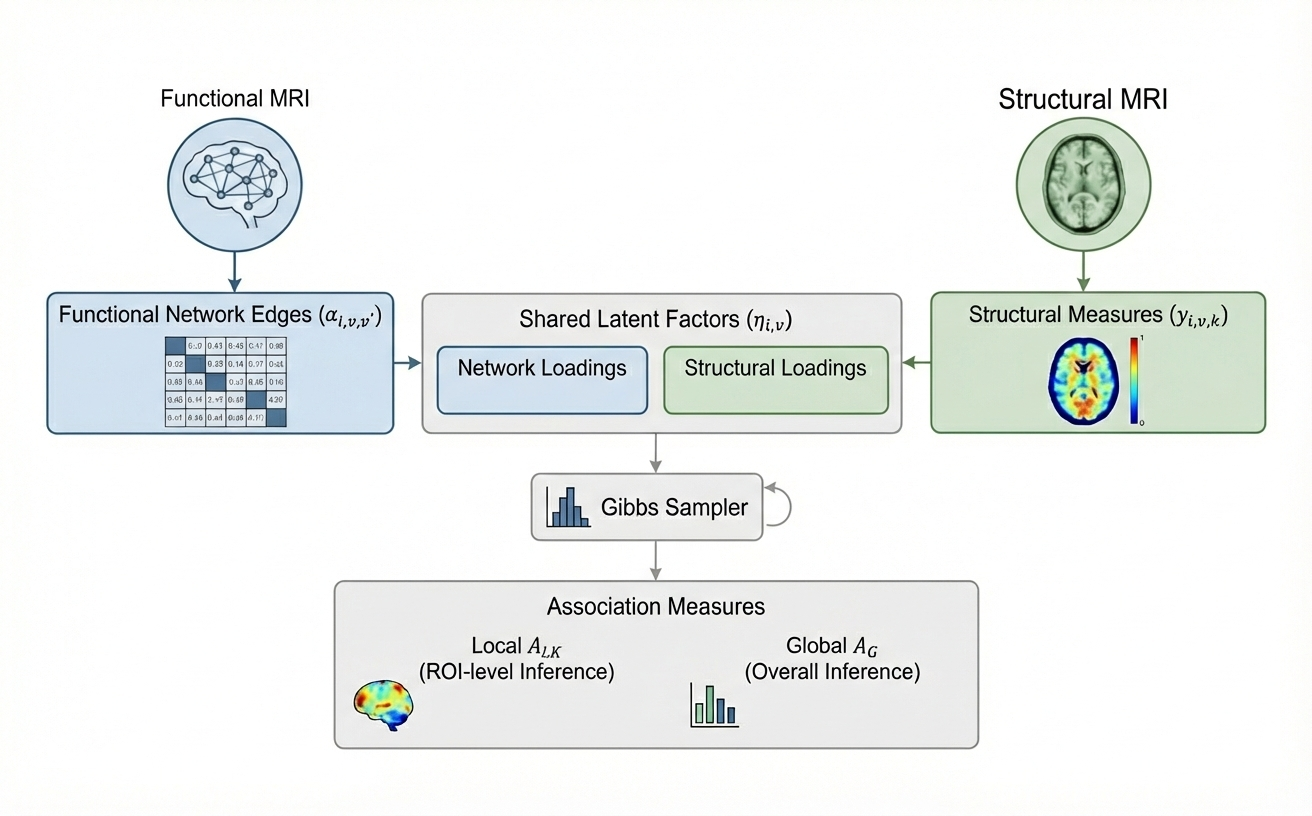}
     \caption{Overview of the proposed Bayesian framework. Functional MRI time series are summarized as functional connectivity networks with edges $a_{i,v,v'}$ (connecting ROIs $\mathcal{R}_v$ and $\mathcal{R}_{v'}$), and structural MRI data provide ROI-wise structural measures $y_{i,v,k}$ at ROI $\mathcal{R}_v$. Both modalities are linked through shared latent factors $\eta_{i,v}$. A Gibbs sampler is used to obtain the joint posterior over all parameters, from which we derive a global association measure $A_G$ (overall coupling across modalities) and local association measures $A_{L,k}$ (modality-specific structure--function coupling).}
    \label{fig:General_workflow_paper}
\end{figure}

\textbf{Global association testing.} Let $\boldsymbol{\Psi}_r$ denote a $P \times P$ diagonal matrix with the $k$-th diagonal entry $\psi_{k,r}$, and let $\delta_{\psi,r} = \det(\mathbf{I} - \boldsymbol{\Psi}_r)$ denote the indicator of the residual structure after excluding latent dimension $r$ from all structural variables. Collecting these across dimensions, define $\boldsymbol{\delta}_\psi = (\delta_{\psi,1}, \dots, \delta_{\psi,R})^T$. The global association metric is then defined as
\begin{equation}
    A_{G}(\boldsymbol{\gamma}, \boldsymbol{\psi}) = \boldsymbol{\gamma}^T (\mathbf{1}_R - \boldsymbol{\delta}_{\psi}).    
\end{equation}
This metric counts the number of latent dimensions that are simultaneously active in the graph model ($\gamma_r = 1$) and in at least one structural variable model ($\psi_{k,r} = 1$, for some $k$). A value of $P(A_G > 0 | \text{data}) = 0$ indicates modality-specific structure with no shared latent factors, while $P(A_G > 0 | \text{data}) > 0$ implies the presence of shared latent dimensions that jointly explain functional connectivity and structural attributes.

\subsection{Identifiability Considerations}

As is typical in latent variable models, the latent factors $\boldsymbol{\eta}_{i,v}$ and $\boldsymbol{\alpha}_{i,k}$ are identifiable only up to orthogonal transformations, and parameters such as $\lambda_r$ and $\omega_{k,r}$ are identifiable only up to sign and scale. However, the $A_G$ and $A_{L,k}$ statistics, depending only on binary inclusion indicators, are invariant to such transformations and remain interpretable. Proposition 1 describes a result confirming this fact.

\newtheorem{proposition}{Proposition}

\begin{proposition}[Invariance of $A_G$ and $A_{L,k}$]
The statistics $A_{L,k}$ and $A_G$ are invariant under signed permutation and scaling transformations.
\end{proposition}

\begin{proof}
Note that, $A_{L,k}=\sum_{r=1}^R \gamma_r\psi_{k,r}$ and 
$A_G=\sum_{r=1}^R\gamma_r[1-\prod_{k=1}^P(1-\psi_{k,r})]$. Therefore, both statistics depend on the zero-nonzero structure of the pairs $(\lambda_r,\omega_{k,r})$. Permutation only relabels coordinates of $\boldsymbol{\gamma}$ and $\boldsymbol{\psi}_{k}$ vectors; since $A_{L,k}$ and $A_G$ are summations over all $r$, their values are invariant to reordering. Similarly, scaling by $c\neq 0$ preserves the zero-nonzero structure of the pair $(\lambda_r,\omega_{k,r})$.

\end{proof}


\subsection{Posterior Computation}
The posterior distributions of the proposed model parameters are analytically intractable, but conditional distributions for each block of parameters are available in closed form due to the conjugate hierarchical prior structure. This enables efficient Gibbs sampling for posterior inference, where parameters are updated iteratively from their full conditional distributions. The explicit forms of these full conditionals are presented in Appendix~A.

After discarding burn-in iterations, we obtain posterior samples $\{\boldsymbol{\gamma}^{(1)}, \dots, \boldsymbol{\gamma}^{(F)}\}$ and $\{\boldsymbol{\psi}_{k}^{(1)}, \dots, \boldsymbol{\psi}_k^{(F)}\}$, $k=1, \dots, P$, for the latent inclusion indicators corresponding to the graph and structural variables. These posterior draws are then used to compute probabilities for both local and global association metrics. The posterior probability that the local association metric $A_{L,k}$ takes a specific value $s$ is estimated by the Monte Carlo proportion of samples for which the realized value equals $s$:
\begin{equation*}
P(A_{L,k}=s \mid \text{data}) = \frac{1}{F} \sum_{f=1}^F I(A_{L,k}(\boldsymbol{\gamma}^{(f)}, \boldsymbol{\psi}_k^{(f)}) = s),    
\end{equation*}
where $I(\cdot)$ denotes the indicator function. This provides a direct probabilistic statement about the strength of association between the graph and the $k$-th structural variable. Similarly, the global association metric $A_{G}$ is evaluated at each posterior draw. The posterior probability of $A_G$ taking a particular value $s$ is given by
\begin{equation*}
P(A_{G}=s \mid \text{data}) = \frac{1}{F} \sum_{f=1}^F I(A_{G}(\boldsymbol{\gamma}^{(f)}, \boldsymbol{\psi}^{(f)}) = s).    
\end{equation*}
This yields an interpretable measure of evidence for global structural–functional association across all modalities.

While the main emphasis of the model is on testing associations 
between structural and functional modalities, the framework also 
supports predictive inference. We note that because the model 
includes subject-specific parameters ($\mu_i$, $\theta_{i,k}$, 
$\boldsymbol{\eta}_{i,v}$, $\boldsymbol{\alpha}_{i,k}$), prediction 
for new (out-of-sample) subjects is not directly available from 
the fitted model. We therefore evaluate predictive performance 
via a node-holdout scheme, in which structural measurements 
$y_{i,v^*,k}$ at a subset of ROIs are masked prior to model 
fitting. For each posterior draw $f=1,\dots,F$, predictive samples 
are generated as 
$N(\theta_{i,k}^{(f)} + \mathbf{x}_i^T 
\boldsymbol{\beta}_{x,S}^{(f)} + 
g(\boldsymbol{\eta}_{i,v^*}^{(f)}, 
\boldsymbol{\alpha}_{i,k}^{(f)}, 
\boldsymbol{\omega}_k^{(f)}), \phi_k^{(f)2})$, 
where the posterior distribution of 
$\boldsymbol{\eta}_{i,v^*}$ is informed by the observed 
functional connectivity edges $a_{i,v^*,v'}$ rather than by 
the masked structural data. Details of this evaluation are 
presented in Section~\ref{sec-simulation}.

\section{Simulation Studies}\label{sec-simulation}

To assess the performance of our proposed Bayesian hierarchical model for the joint analysis of brain connectivity graphs and structural data, we carried out a comprehensive simulation study under a range of settings. The primary goal was to evaluate the model’s ability to accurately estimate both the local association metrics $\{A_{L,k}:k=1,...,P\}$ and the global association metric $A_{G}$, which together quantify the dependence between graph connectivity and structural measurements. In addition, we examined the model’s predictive accuracy in terms of offering point prediction and predictive uncertainty for graph edges and structural data, and compare the performance with marginal models that analyze graph and structural data separately.

\subsection{Data Simulation and Evaluation Criteria}
In each simulation scenario, we generated synthetic datasets according to the graph and structural models in Eq.~\eqref{eq:network_model}, \eqref{eq:structural_model}, and \eqref{eq:bil}, with $N = 200$ subjects and $V = 25$ brain regions. Latent factors $\bm{\eta}_{i,v}$, structural variable specific factors $\bm{\alpha}_{i,k}$, structural variable specific intercept $\theta_{i,k}$, and subject-specific intercept for the graph model $\mu_i$ were sampled from their respective prior distributions, setting the dimension of both $\bm{\eta}_{i,v}$ and $\bm{\alpha}_{i,k}$ to $R^*=5$. To systematically evaluate model performance, we varied several key parameters:
\begin{itemize}
\item \textbf{Number of structural variables:} $P\in\{3,5\}$.
\item \textbf{True global association:} $A_{G}^*\in\{0,2,4\}$, where $0$ indicates no global association between structural and graph data, with stronger dependence at larger values.
\item \textbf{Signal-to-noise ratio (SNR):} To examine robustness under different data quality conditions, we varied the SNR for both network and structural components, defined as
\[
\text{SNR}_{\text{graph}} = \frac{\text{Var}(\boldsymbol{\eta}_{i,v}^\top \boldsymbol{\Delta} \boldsymbol{\eta}_{i,v'})}{\sigma^2}, \quad
\text{SNR}_{\text{structural},k} = \frac{\text{Var}(\boldsymbol{\eta}_{i,v}^\top \boldsymbol{\Omega}_k \boldsymbol{\alpha}_{i,k})}{\phi_k^2},
\]
where $\boldsymbol{\Delta} = \mathrm{diag}(\lambda_1, \dots, \lambda_R)$ and $\boldsymbol{\Omega}_k = \mathrm{diag}(\omega_{k,1}, \dots, \omega_{k,R})$.
Datasets were generated under three regimes: High (SNR 4–8), Moderate (SNR 1–2), and Very Low (SNR 0.2–0.5) by scaling the noise variances while holding the signal magnitudes fixed.
\end{itemize}

This simulation framework encompasses six distinct settings, each examined under high, moderate, and very low signal-to-noise ratio (SNR) conditions. In order to mirror the signal characteristics typically observed in neuroimaging studies, we emphasize moderate SNR values, which represent the most common and informative regime, while avoiding scenarios with unrealistically high SNR. Very low SNR settings are included primarily to identify the point at which model performance deteriorates, though we do not expect the model to capture meaningful associations in these cases.

The simulations vary systematically in the number of structural variables, the extent of cross-modality associations, and the level of SNR. This comprehensive design allows us to assess the joint model’s ability to recover local and global association metrics across a spectrum of realistic data-generating conditions. For each simulation scenario, the joint model is fit using latent dimensionalities $R=5$ and $R=10$, enabling evaluation of how the choice of latent space size influences model performance. Combining six foundational data settings with two choices of $R$ results in a total of $12$ simulation scenarios, as detailed in Table~\ref{tab:simsettings}. To further ensure the robustness and reproducibility of our results, each scenario is replicated $20$ times.

\begin{table}[htbp]
\centering
\caption{Summary of the 12 simulation scenarios. Each scenario is evaluated under three SNR regimes (high, moderate, and very-low).}
\label{tab:simsettings}
\begin{tabular}{@{}ccccc | ccccc@{}}
\toprule
\multicolumn{5}{c|}{\textbf{Scenarios 1--6} ($P=3$)} & \multicolumn{5}{c}{\textbf{Scenarios 7--12} ($P=5$)} \\
\midrule
\textbf{Sim} & $\boldsymbol{P}$ & $\boldsymbol{R_{\text{fit}}}$ & $\boldsymbol{A_G^*}$ & $\boldsymbol{A_{L,k}^*}$ & \textbf{Sim} & $\boldsymbol{P}$ & $\boldsymbol{R_{\text{fit}}}$ & $\boldsymbol{A_G^*}$ & $\boldsymbol{A_{L,k}^*}$ \\
\midrule
1 & 3 & 5  & 0 & (0, 0, 0) & 7  & 5 & 5  & 0 & (0, 0, 0, 0, 0) \\
2 & 3 & 5  & 2 & (2, 2, 2) & 8  & 5 & 5  & 2 & (2, 2, 2, 2, 2) \\
3 & 3 & 5  & 4 & (4, 4, 4) & 9  & 5 & 5  & 4 & (4, 4, 4, 4, 4) \\
4 & 3 & 10 & 0 & (0, 0, 0) & 10 & 5 & 10 & 0 & (0, 0, 0, 0, 0) \\
5 & 3 & 10 & 2 & (2, 2, 2) & 11 & 5 & 10 & 2 & (2, 2, 2, 2, 2) \\
6 & 3 & 10 & 4 & (4, 4, 4) & 12 & 5 & 10 & 4 & (4, 4, 4, 4, 4) \\
\bottomrule
\end{tabular}
\end{table}

We estimated the posterior distribution of the global association metric $A_G$ and visualized these results using dotplots. Due to the large number of local association metrics $A_{L,k}$ for $k=1,...,P$, it is impractical to plot each individually. To succinctly summarize the posterior distributions of $A_{L,k}$, we assess estimation accuracy using the absolute deviation (AD)
\begin{align}
AD(A_{L,k})=|E(A_{L,k})-A_{L,k}^*|,\:\:\mbox{where}\:\:E(A_{L,k})=\sum_{r=0}^R r*P(A_{L,k}=r)    
\end{align}
Here, $E(A_{L,k})$  denotes the posterior mean of $A_{L,k}$, and $A_{L,k}^*$
is the true value under the data-generating process. This metric reflects the difference between the estimated posterior mean and the actual local association for each $k=1,...,P$. By jointly evaluating posterior probabilities of $A_G$ and absolute deviations $AD(A_{L,k})$, we systematically characterize the posterior distributions of both global and local association metrics, providing a rigorous assessment of the model's ability to recover true underlying associations.

To further evaluate the advantage of jointly modeling graph and structural data, we compared our proposed approach with a competing baseline strategy in which the structural modality is analyzed independently. Specifically, we compared the joint model against a separate structural-only model, which is estimated in isolation without sharing latent factors with the graph data. Using the node-holdout scheme described in Section~\ref{sec-simulation}, we then compared the mean squared predictive error (MSPE) for the masked structural nodes, and quantified imputation uncertainty through the coverage and length of 95\% predictive intervals. This comparison highlights how exploiting shared structure across modalities in the joint model can substantially improve node-level imputation accuracy. Because the joint model allows the observed edges in the graph modality to inform the missing structural measurements, it yields superior recovery of the masked data compared to the structural-only model, particularly when the degree of association between the two modalities is strong.

\subsection{Study of Global and Local Association Metrics}

Figure~\ref{fig:your_label_high} presents the estimated posterior probabilities $P(A_G=r)$ for each possible value $r=1,..,R$ across various simulation scenarios in the high SNR setting. The results reveal a consistently clear pattern: the posterior probability corresponding to the true value, $P(A_G=A_G^*)$, is overwhelmingly dominant and approaches $1$ in every scenario. This reflects both the model’s high confidence and precise accuracy in detecting the true global association strength when the signal is strong. Notably, $P(A_G=A_G^*)$ serves as a direct measure of uncertainty, indicating that the model reliably and decisively identifies the true value with minimal ambiguity in its inference.
Furthermore, the model’s performance under high SNR demonstrates robustness; it remains unaffected by the total number of structural variables included in the data, as well as by the choice of latent dimensionality ($R=5$ or $R=10$). This resilience suggests that the method is stable and reliable across a broad range of realistic neuroimaging conditions.

\begin{figure}[H]
    \centering
    \includegraphics[width=\textwidth]{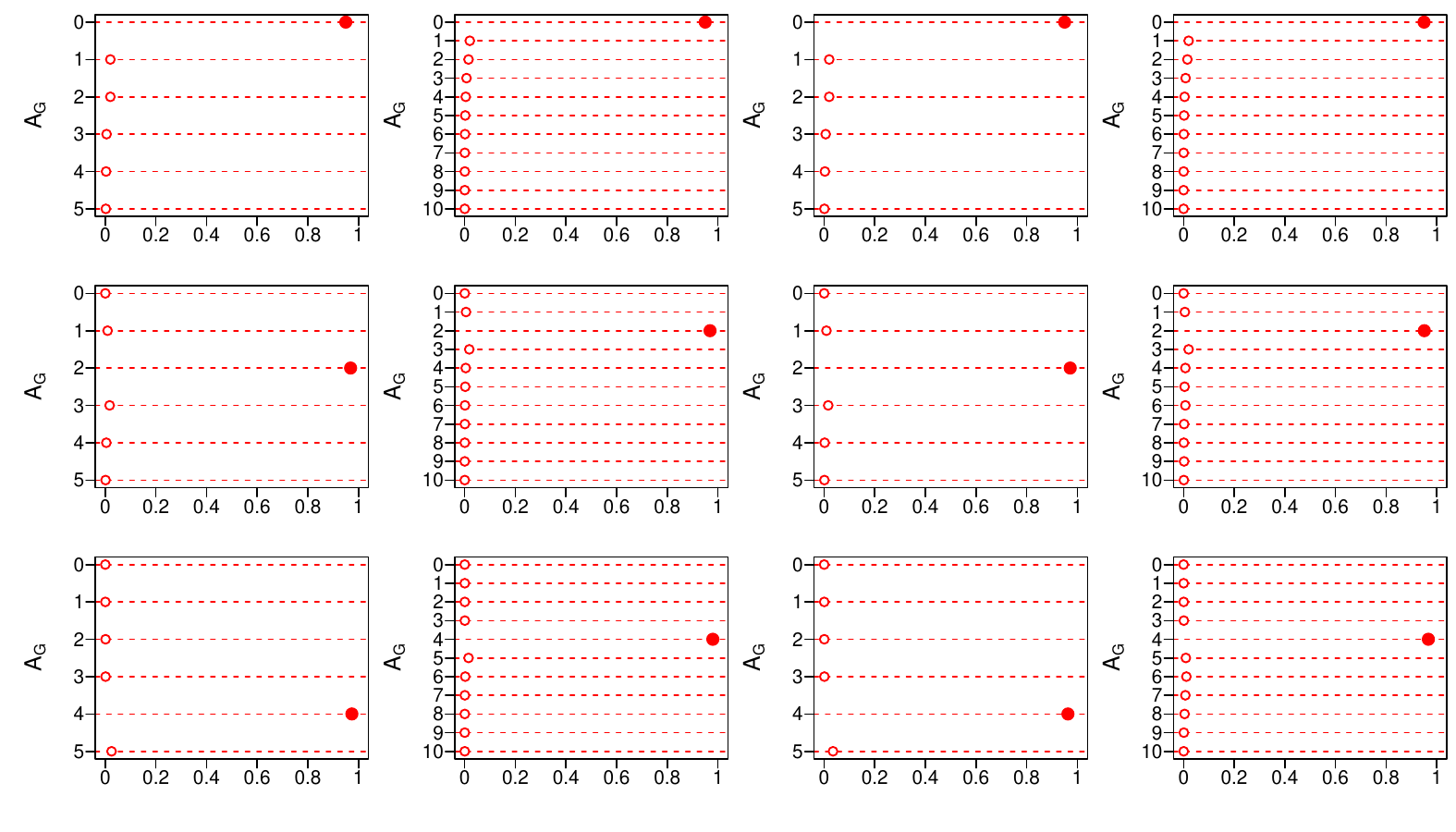}
    \caption{Dotplots present the posterior distributions of $A_G$ across the twelve simulation scenarios under high SNR conditions. On the y-axis, the red point indicates the true value of $A_G$. The layout assigns the first column to scenarios 1–3, the second to 4–6, the third to 7–9, and the fourth to 10–12. Across all scenarios, the posterior probability that $A_G$ matches its true value is close to 1, underscoring the robust accuracy of the proposed method.}
    \label{fig:your_label_high}
\end{figure}

A similar pattern emerges under moderate SNR, as illustrated in Figure~\ref{fig:your_label_mid}. Here too, $P(A_G=A_G^*)$ is maximized, and model performance is unaffected by either the complexity of the data or the latent space size. These findings are particularly relevant given that high and moderate SNR settings closely match the signal quality encountered in typical neuroimaging studies.

\begin{figure}[H]
    \centering
    \includegraphics[width=\textwidth]
    {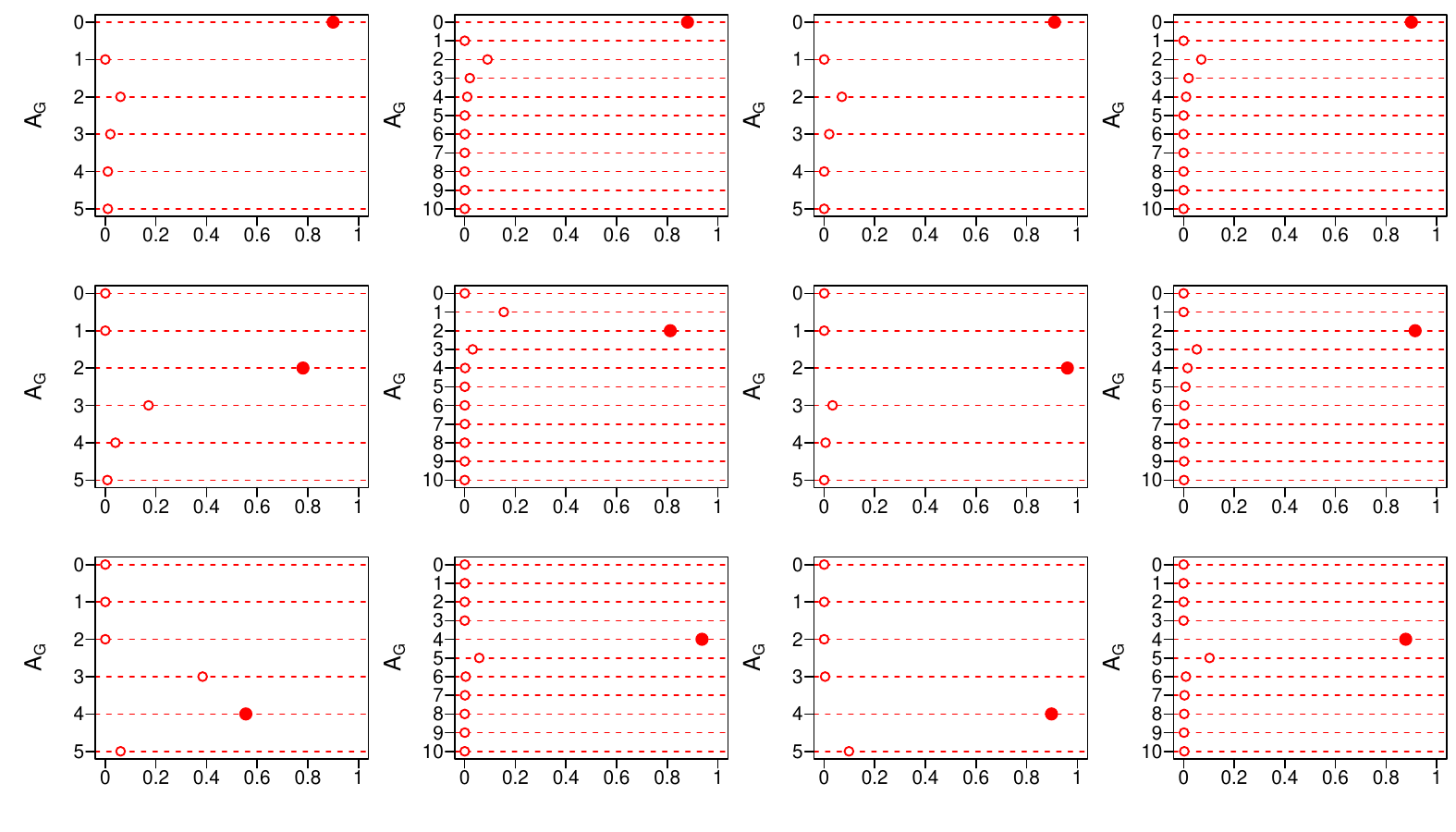}
    \caption{Dotplots present the posterior distributions of $A_G$ across the twelve simulation scenarios under moderate SNR conditions. On the y-axis, the red point indicates the true value of $A_G$. The layout assigns the first column to scenarios 1–3, the second to 4–6, the third to 7–9, and the fourth to 10–12. Across all scenarios, the posterior probability that $A_G$ matches its true value is close to 1, underscoring the robust accuracy of the proposed method.}
    \label{fig:your_label_mid}
\end{figure}

As the signal-to-noise ratio decreases, particularly in the very low SNR setting, the model performs well by correctly signaling the absence of a global association between the connectivity graph and structural variables, as evidenced in the first row of Figure~\ref{fig:your_label_low}. However, when the true strength of association increases (i.e., $A_G^*=2$ or $4$), the model’s ability to recover the correct association deteriorates markedly. In these scenarios, the posterior probability assigned to incorrect values, $P(A_G\neq A_G^*)$, exceeds that assigned to the true value, highlighting a substantial decline in inference accuracy caused by the limited signal.

\begin{figure}[H]
    \centering
    \includegraphics[width=\textwidth]{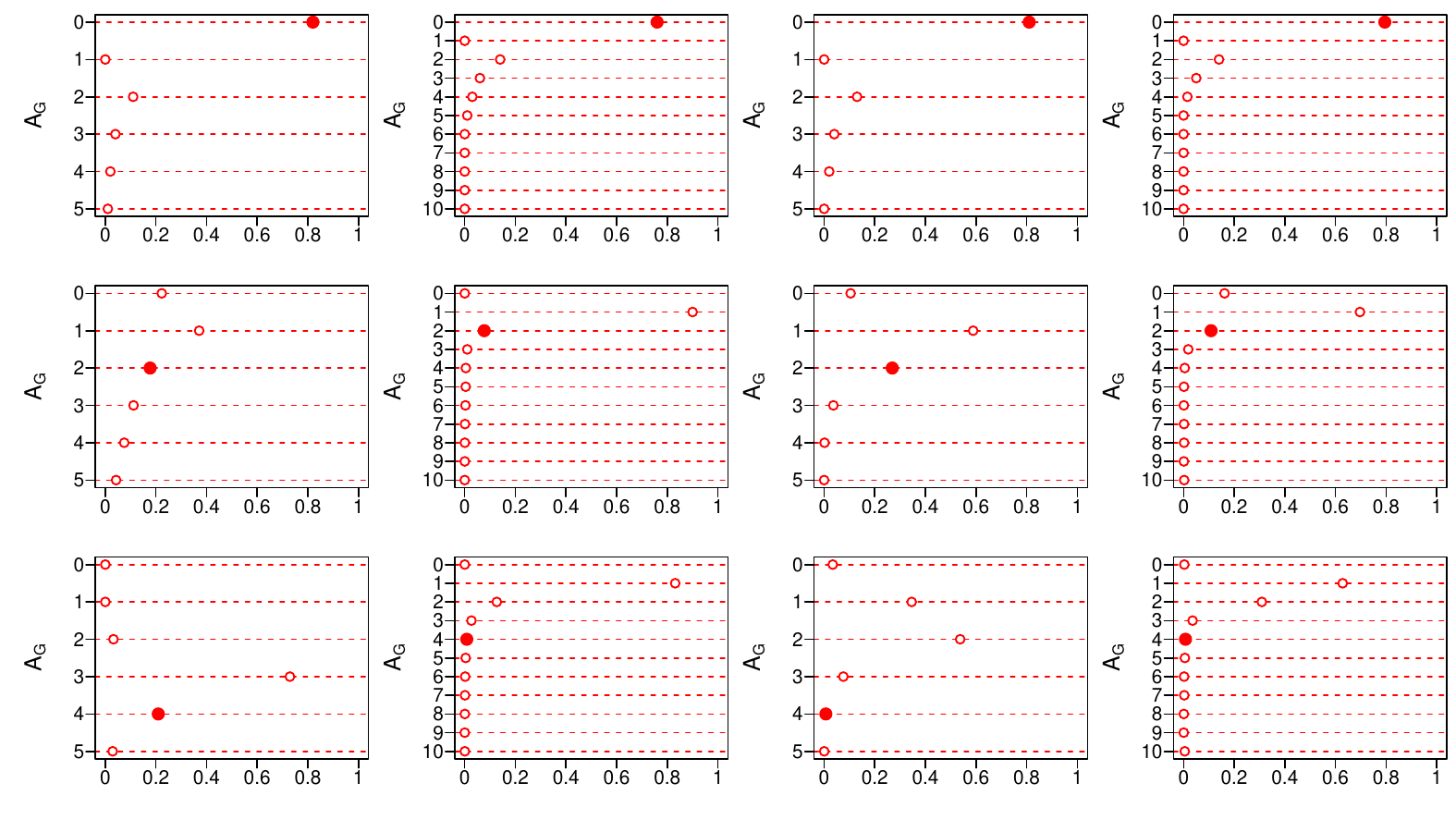}
    \caption{Dotplots display the posterior distributions of $A_G$ for each of the twelve simulation scenarios under very low SNR conditions. The red point on the y-axis marks the true value of $A_G$. The arrangement places scenarios 1–3 in the first column, 4–6 in the second, 7–9 in the third, and 10–12 in the fourth. In these low-signal settings, the method has difficulty assigning the highest posterior probability to the true value of $A_G$  when $A_G^*$ is nonzero, that is, as the true strength of association increases, the model’s recovery performance diminishes.}
    \label{fig:your_label_low}
\end{figure}

Moreover, this decline in performance is further exacerbated as the latent dimensionality parameter $R$ increases, despite the true value being $R^*=5$. Under very low SNR, the model struggles even more with higher latent space dimensions, suggesting that excessive model complexity amplifies estimation uncertainty and reduces reliability in the presence of weak signal. These findings emphasize the challenges inherent in association recovery under adverse signal conditions and underscore the importance of aligning model complexity with available data quality.

\begingroup
\setlength{\tabcolsep}{3pt} 
\footnotesize 
\begin{longtable}{@{}lcccccc@{}}
\caption{Absolute deviation of local associations AD($A_{LK}$) across Low, Medium, and High SNR levels. Each cell is $P$-dimensional showing $AD(A_{L,k})$ values for $k=1,..,P$. Smaller values indicate better recovery of the true local association. The model shows recovery with high accuracy for most simulation scenarios. }
\label{tab:local_association_results} \\
\toprule
Sim & $P$ & $R_{\text{fitted}}$ & $A_{L,k}^*$ & Low SNR: AD($A_{L,k}$) & Med SNR: AD($A_{L,k}$) & High SNR: AD($A_{L,k}$) \\
\midrule
\endfirsthead
\caption[]{Absolute deviation of local associations (continued)} \\
\toprule
Sim & $P$ & $R_{\text{fitted}}$ & $A_{L,k}^*$ & Low & Med & High \\
\midrule
\endhead
\bottomrule
\endlastfoot
1  & 3 & 5  & (0,0,0) & (0.81, 0.32, 0.21) & (0.28, 0.11, 0.12) & (0.01, 0.02, 0.03) \\
2  & 3 & 5  & (2,2,2) & (1.91, 1.85, 0.58) & (0.31, 0.23, 0.83) & (0.02, 0.21, 0.13) \\
3  & 3 & 5  & (4,4,4) & (3.90, 3.95, 3.96) & (0.91, 0.35, 0.91) & (0.10, 0.20, 0.19) \\
4  & 3 & 10 & (0,0,0) & (0.71, 0.67, 0.91) & (0.21, 0.12, 0.19) & (0.01, 0.02, 0.01) \\
5  & 3 & 10 & (2,2,2) & (1.97, 1.97, 1.97) & (1.12, 0.55, 0.97) & (0.12, 0.11, 0.20) \\
6  & 3 & 10 & (4,4,4) & (3.00, 2.99, 3.90) & (2.75, 1.20, 1.33) & (0.09, 0.06, 0.02) \\
7  & 5 & 5  & (0,0,0,0,0) & (1.12, 2.13, 1.11, 0.99, 1.31) & (0.91, 0.88, 1.23, 0.78, 0.56) & (0.02, 0.09, 0.11, 0.08, 0.02) \\
8  & 5 & 5  & (2,2,2,2,2) & (1.63, 2.86, 2.78, 2.11, 2.12) & (1.01, 1.01, 0.12, 0.11, 0.98) & (0.01, 0.02, 0.01, 0.01, 0.01) \\
9  & 5 & 5  & (4,4,4,4,4) & (3.20, 2.10, 3.89, 0.82, 0.92) & (1.15, 0.36, 0.71, 0.91, 0.21) & (0.01, 0.21, 0.01, 0.09, 0.12) \\
10 & 5 & 10 & (0,0,0,0,0) & (1.13, 2.22, 1.29, 1.12, 2.21) & (0.82, 0.15, 0.25, 0.13, 0.11) & (0.01, 0.01, 0.09, 0.08, 0.01) \\
11 & 5 & 10 & (2,2,2,2,2) & (1.99, 1.95, 1.98, 1.78, 1.89) & (1.01, 1.01, 1.57, 0.89, 0.93) & (0.01, 0.06, 0.12, 0.31, 0.01) \\
12 & 5 & 10 & (4,4,4,4,4) & (3.40, 2.76, 3.90, 2.19, 3.11) & (0.89, 0.21, 1.10, 0.98, 0.99) & (0.01, 0.12, 0.09, 0.98, 0.21) \\
\end{longtable}
\endgroup

Table~\ref{tab:local_association_results} presents the accuracy of estimating the local association metric $A_{L,k}$ for each $k=1,...,P$. In high SNR scenarios, the accuracy metric remains nearly zero, indicating highly precise recovery of the true local association values $A_{L,k}^*$. As the SNR decreases, estimation accuracy declines. Moderate SNR regimes still yield reasonable identification of the true local associations, while very low SNR conditions result in noticeably elevated accuracy scores, reflecting increased estimation error.

Under very low SNR, this deterioration becomes increasingly pronounced as the true local association strength grows (i.e., with larger $A_{L,k}^*$). Furthermore, estimation challenges in these low-signal settings are exacerbated when the latent dimensionality is misspecified (e.g., when $R=10$ while the ground truth is $R^*=5$), leading to higher errors in estimating $A_{L,k}^*$. It is important to note that such latent space misspecification has a negligible effect in moderate and high SNR scenarios, where the model’s recovery of local associations remains robust.


\subsection{Predictive Inference via Node-Holdout Imputation}

As described in Section~\ref{sec-simulation}, we evaluate predictive performance by masking structural measurements $y_{i,v^*,k}$ at a random subset of nodes for each subject and comparing the joint model against a structural-only baseline. Under the joint model, the shared latent factor $\eta_{i,v^*}$ is informed by the observed functional connectivity edges $a_{i,v^*,v'}$, whereas the separate model must rely entirely on the prior.

If the modalities are globally independent (i.e., $A_G^*=0$), both the joint and separate models should yield comparable mean squared prediction errors (MSPEs). However, when an association is present ($A_G^* > 0$), the joint model should achieve better predictive performance, along with improved characterization of predictive uncertainty. Accordingly, we exclude simulation settings with $A_G^*=0$ (namely, scenarios 1, 4, 7 and 10), and focus our comparison on scenarios where $A_G^*>0$, specifically, scenarios 2, 3, 5, 6, 8, 9, 11 and 12, across high, moderate, and very-low SNR levels.

\begin{figure}[htbp]
    \centering
    \begin{subfigure}[b]{0.32\textwidth}
        \includegraphics[width=\linewidth]{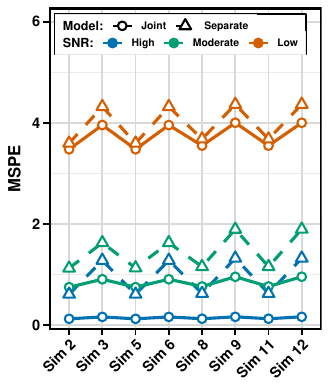}
        \caption{MSPE}
    \end{subfigure}
    \hfill
    \begin{subfigure}[b]{0.32\textwidth}
        \includegraphics[width=\linewidth]{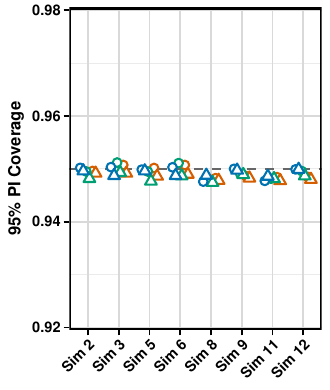}
        \caption{95 \% PI Coverage}
    \end{subfigure}
    \hfill
    \begin{subfigure}[b]{0.32\textwidth}
        \includegraphics[width=\linewidth]{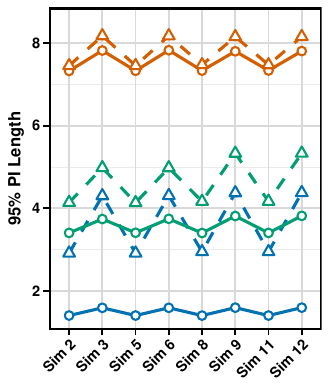}
        \caption{95 \% PI Length}
    \end{subfigure}
    \caption{Structural node-holdout imputation: (a) MSPE, (b) 95\% PI Coverage, and (c) 95\% PI Length for Joint and Separate models across associated scenarios ($A_G^* > 0$), evaluated under high (blue), moderate (green), and low (orange) SNR levels. Results are averaged over 20 replications.}
    \label{fig:prediction_trends}
\end{figure}

Figure~\ref{fig:prediction_trends} illustrates the average predictive performance for structural node-holdout imputation across varying SNR regimes. The joint model consistently outperforms the separate structural-only model, achieving a lower MSPE and significantly narrower predictive intervals across all simulated scenarios. This performance gap is a direct consequence of the joint model's ability to leverage the observed functional connectivity edges ($a_{i,v^*,v'}$) to inform the posterior distribution of the shared latent factor $\eta_{i,v^*}$ for the held-out node. In contrast, the separate model, lacking access to the network modality, is forced to rely on the hierarchical prior distribution for the held-out latent factors, $\boldsymbol{\eta}_{i,v^*} \sim N(\mathbf{0}, \sigma_\eta^2 \mathbf{I})$. Furthermore, as shown in Figure~\ref{fig:prediction_trends}b, the joint model maintains nominal 95\% coverage even while providing significantly greater predictive precision. This indicates that the reduction in interval length (Figure~\ref{fig:prediction_trends}c) is not achieved through overconfidence, but rather through a genuine gain in information afforded by the cross-modality borrowing of strength. While the predictive advantage is most pronounced in high and moderate SNR conditions, the joint model remains superior even in low SNR settings, where the network topology provides a critical stabilizing signal that is otherwise absent in the structural-only approach.

\begin{figure}[htbp]
    \centering
    \begin{subfigure}[b]{0.32\textwidth}
        \includegraphics[width=\linewidth]{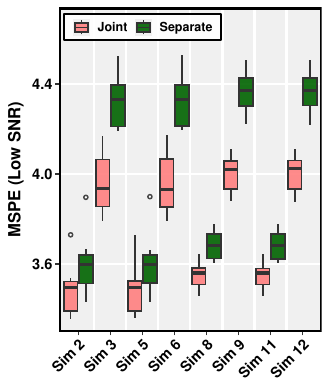}
        \caption{Low SNR}
    \end{subfigure}
    \hfill
    \begin{subfigure}[b]{0.32\textwidth}
        \includegraphics[width=\linewidth]{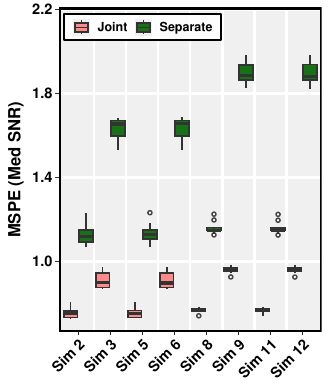}
        \caption{Med SNR}
    \end{subfigure}
    \hfill
    \begin{subfigure}[b]{0.32\textwidth}
        \includegraphics[width=\linewidth]{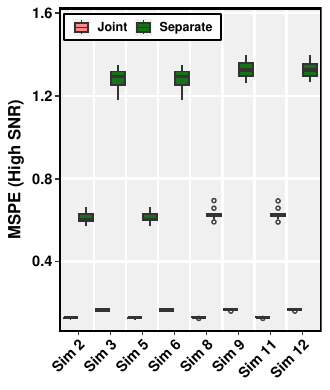}
        \caption{High SNR}
    \end{subfigure}
    \caption{Distribution of MSPE for structural node-holdout imputation across 20 replications under (a) low, (b) medium, and (c) high SNR. Boxplots compare the Joint model (pink) and Separate model (green) across associated scenarios ($A_G^* > 0$).}
    \label{fig:mspe_boxplots}
\end{figure}

The distributional characteristics of the prediction error and interval efficiency are further detailed in Figures~\ref{fig:mspe_boxplots} and~\ref{fig:length_boxplots}. As demonstrated by the MSPE boxplots in Figure~\ref{fig:mspe_boxplots}, the joint model consistently yields a lower median error and significantly reduced dispersion compared to the separate model. This reduction in variance across replications suggests that the joint framework provides a more robust and stable imputation mechanism; whereas the separate model's predictions are highly sensitive to the sparsity of the structural data, the joint model regularizes the imputation of missing nodes through the subject-specific connectivity structure. This gain in predictive efficiency is most evident in Figure~\ref{fig:length_boxplots}, where the 95\% PI lengths for the joint model are substantially tighter than those of the separate model across all SNR levels. In the high SNR regime, the separate model exhibits nearly double the interval width of the joint model in several scenarios, reflecting a high degree of predictive waste caused by ignoring the network-informed latent structure. Even as the signal-to-noise ratio decreases, the joint model maintains its advantage in uncertainty quantification. These results collectively demonstrate that by integrating multimodal information into a shared latent space, the joint model achieves a more concentrated posterior for missing structural values, resulting in imputations that are both more accurate and more precise than those produced by isolated modeling.

\begin{figure}[H]
    \centering
    \begin{subfigure}[b]{0.32\textwidth}
        \includegraphics[width=\linewidth]{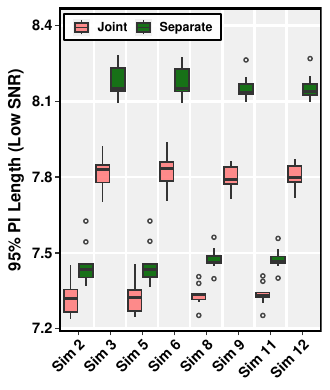}
        \caption{Low SNR}
    \end{subfigure}
    \hfill
    \begin{subfigure}[b]{0.32\textwidth}
        \includegraphics[width=\linewidth]{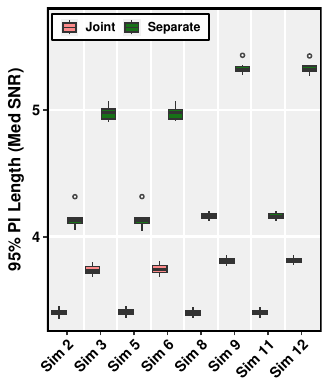}
        \caption{Med SNR}
    \end{subfigure}
    \hfill
    \begin{subfigure}[b]{0.32\textwidth}
        \includegraphics[width=\linewidth]{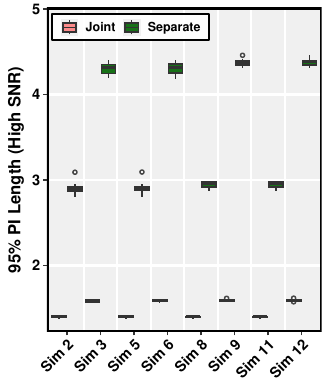}
        \caption{High SNR}
    \end{subfigure}
    \caption{Distribution of 95\% PI lengths for structural node-holdout imputation across 20 replications under (a) low, (b) medium, and (c) high SNR. Joint model (pink) and Separate model (green) are compared across associated scenarios ($A_G^* > 0$).}
    \label{fig:length_boxplots}
\end{figure}


\section{Application to Multi-Modal Brain Imaging Data}\label{sec-application}

We apply the proposed joint modeling framework to brain structural and functional data from
$N = 126$ participants in order to assess the strength of association between the two imaging modalities, after accounting for subject-level covariates. The functional component consists of network-valued data derived from BOLD fMRI time series across 69 ROIs, leading to an undirected graph with $69$ nodes, represented by a $69\times 69$ symmetric adjacency matrix. The structural
component comprises brain volume of each ROI extracted from T1- and T2-weighted scans. For each subject, we also include the APS Score as a behavioral predictor.

We fit the joint model under several choices of the latent dimension, $R \in \{5,10, 15\}$, and compare its performance to that of a corresponding ``solo'' structural model that treats the modality in isolation. For each specification, we summarize posterior inference via (i) the posterior probability of the global association metric being nonzero, $P(A_G>0|\mbox{Data})$, and (ii) predictive performance for the structural data using the node-holdout imputation framework. Predictive performance is evaluated using the mean squared prediction error (MSPE), 95\% predictive interval (PI) coverage, and the average length of the 95\% PIs. Notably, since we have only one structural measure ($P=1$), the local association metric $A_{L,1}$ coincides with the 
global association metric $A_G$, and we therefore report only $A_G$.

Dotplots in Figure~\ref{fig:AG_dotplots} summarize the posterior distribution of the global association
metric \(A_G\) for different choices of the latent dimension \(R\). For each value of \(R\), we obtain
\(P(A_G>0) \approx 1\), indicating overwhelming evidence for a nontrivial association between the
structural and functional modalities. Moreover, across all values of \(R\), the posterior mass is
concentrated at \(A_G = 3\), so that \(P(A_G = 3 \mid \text{Data})\) is the largest among the
possible association levels.

We also observe that \(P(A_G = 3 \mid \text{Data})\) increases with the latent dimension, yielding
the strongest posterior support for \(A_G = 3\) when \(R = 15\). This behavior is expected:
a larger \(R\) allows the model to capture more complex shared variation between the structural and functional data by accommodating a richer low-rank dependence structure. As \(R\) increases, the
model can more flexibly represent cross-modal covariance patterns, leading to a clearer separation between weaker and stronger association regimes. Consequently, the posterior more confidently assigns probability to the highest association level (\(A_G = 3\)) in higher-dimensional latent spaces,
reflecting an enhanced ability to detect and characterize the underlying structure--function coupling. The result stabilizes when $R$ increases above $15$, and is not shown.

\begin{figure}[htbp]
    \centering
    \subfloat[$R = 5$]{%
        \includegraphics[width=0.31\textwidth]{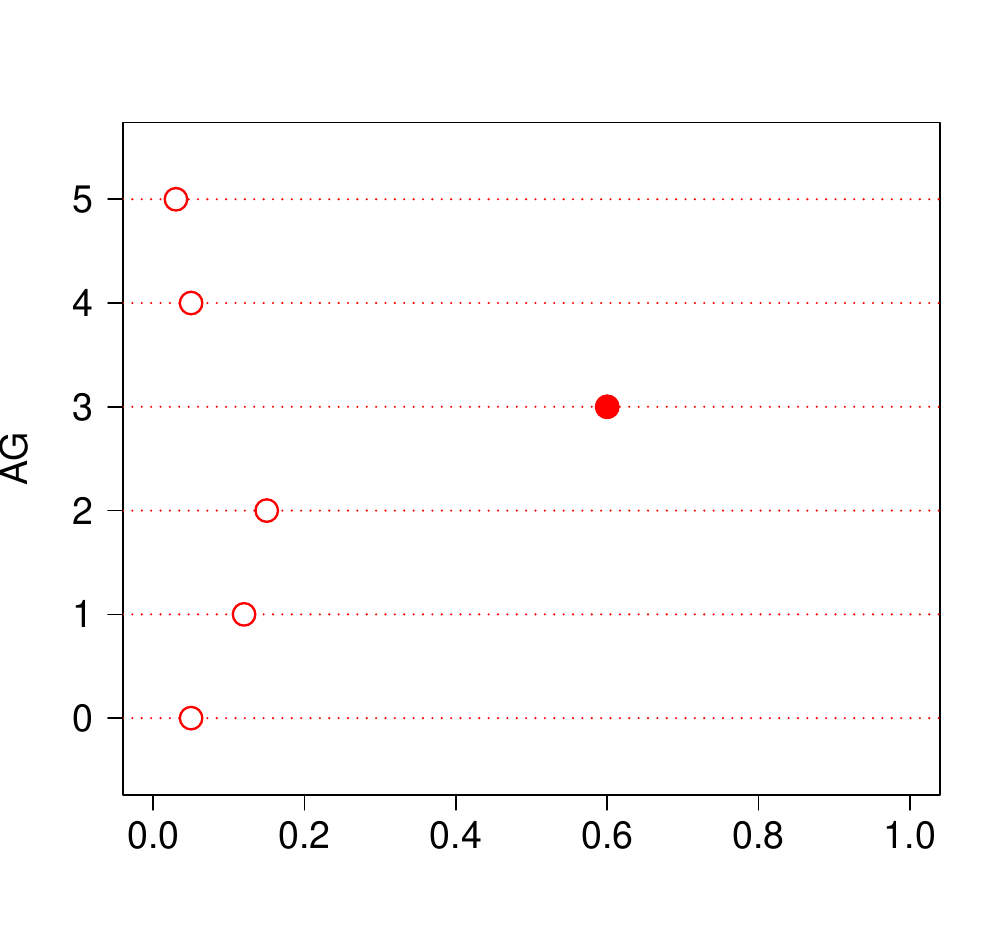}%
    }\hfill
    \subfloat[$R = 10$]{%
        \includegraphics[width=0.31\textwidth]{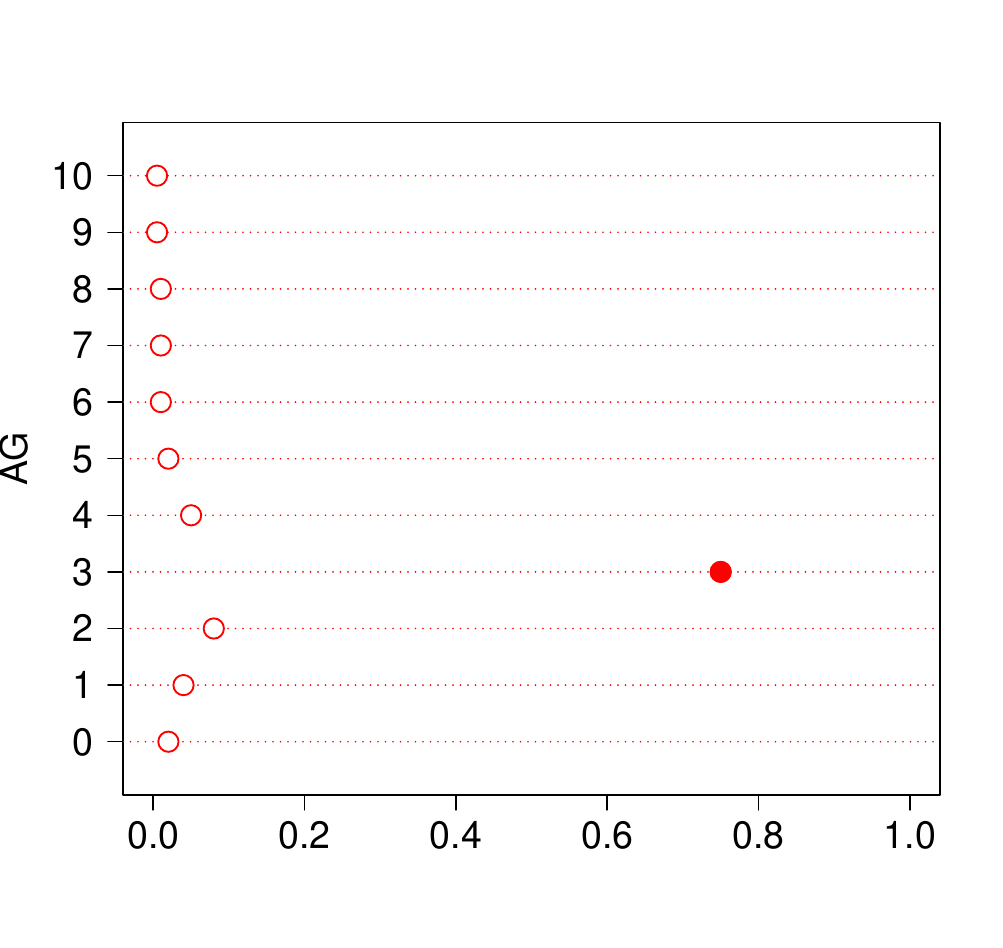}%
    }\hfill
    \subfloat[$R = 15$]{%
        \includegraphics[width=0.31\textwidth]{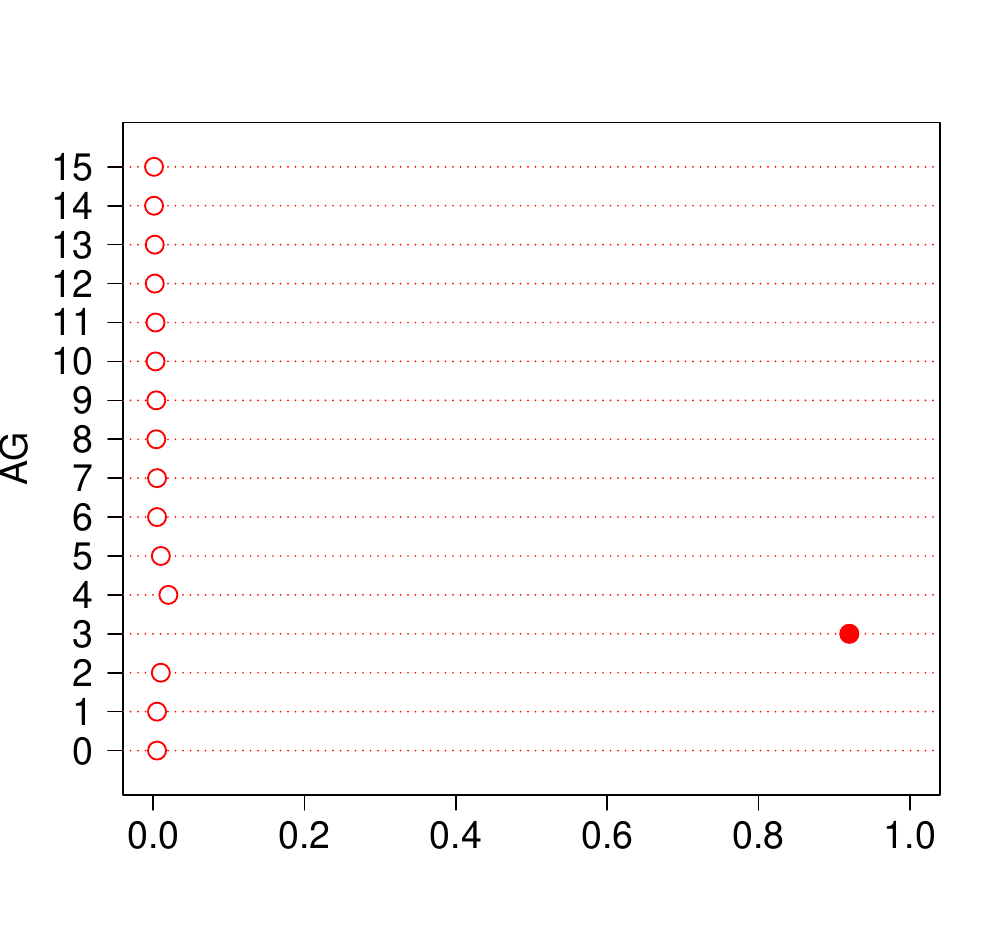}%
    }
    \caption{Dotplots of the global association indicator $A_G$ across latent dimensions $R$. The result shows $P(A_G=3)$ is high consistently for different values of $R$ indicating reliable inference in identifying strong structural-functional data association.}
    \label{fig:AG_dotplots}
\end{figure}

Table \ref{tab:real_results} presents the node-holdout predictive performance on real neuroimaging data across varying latent dimensions. The Joint model uniformly outperforms the Separate approach across all evaluated ranks, demonstrating the clear advantage of leveraging shared cross-modal information for network imputation. While the Separate model's predictive error and uncertainty metrics remain stagnant regardless of the latent dimension, the Joint framework exhibits sensitivity to the chosen rank. It achieves its optimal point-prediction accuracy at the lowest evaluated rank, with performance slightly stabilizing at higher dimensions. Crucially, irrespective of the latent dimension, the Joint model provides superior uncertainty quantification, consistently yielding narrower predictive intervals while maintaining empirical coverage closer to the nominal target.

\begin{longtable}[]{@{}lcccccc@{}}
\caption{Node-holdout predictive performance on real neuroimaging data across latent dimensions $R \in \{5, 10, 15\}$. Bold values indicate the superior model.}\label{tab:real_results}\tabularnewline
\toprule\noalign{}
 & \multicolumn{2}{c}{$R = 5$} & \multicolumn{2}{c}{$R = 10$} & \multicolumn{2}{c}{$R = 15$} \\
Metric & Joint & Separate & Joint & Separate & Joint & Separate \\
\midrule\noalign{}
\endfirsthead
\toprule\noalign{}
 & \multicolumn{2}{c}{$R = 5$} & \multicolumn{2}{c}{$R = 10$} & \multicolumn{2}{c}{$R = 15$} \\
Metric & Joint & Separate & Joint & Separate & Joint & Separate \\
\midrule\noalign{}
\endhead
\bottomrule\noalign{}
\endlastfoot
MSPE     & \textbf{0.512} & 0.628 & \textbf{0.567} & 0.628 & \textbf{0.568} & 0.629 \\
Length   & \textbf{2.835} & 2.997 & \textbf{2.880} & 2.998 & \textbf{2.847} & 2.998 \\
Coverage & \textbf{0.948} & 0.941 & \textbf{0.949} & 0.939 & \textbf{0.945} & 0.941 \\
\end{longtable}


\section{Conclusion and Future Work}\label{sec-conclusion}
In this paper, we proposed a Bayesian joint modeling framework for quantifying the association between functional neuroimaging data, expressed in the form of a brain network over ROIs, and multivariate structural measures over ROIs, with an application to multimodal neuroimaging data. Our approach is designed to test for structure--function association, while flexibly borrowing information across modalities through a shared low-rank latent representation. The model accommodates subject-level behavioral covariates, such as the Aggregate Pegboard score, and yields interpretable probability-of-association summaries that directly quantify evidence for structure–function coupling. Empirical results from the neuroimaging application indicate overwhelming evidence for a nontrivial association between structural and functional modalities, which appears to be robust to the choice of model tuning parameters (e.g., dimensions of the latent factors), reflecting the model's enhanced ability to capture complex shared variation in richer latent spaces. Comparisons between joint modeling of structural and functional data, and solo models further demonstrate that joint modeling improves predictive performance for both network and structural outcomes, as evidenced by lower mean squared prediction error and shorter or comparable predictive intervals.  

Several promising directions for future research emerge from this work. First, on the methodological front, it would be valuable to extend the proposed framework to include additional imaging modalities (e.g., diffusion MRI, task-based fMRI, or molecular imaging), enabling comprehensive multi-modal integration within a unified probabilistic model and allowing for pairwise association testing between imaging modalities. Second, broadening the current cross-sectional approach to longitudinal or developmental data would make it possible to study how structure-function associations change over time or during disease progression. Lastly, there is a need to develop scalable inference algorithms suited for very high-dimensional networks and large cohorts, potentially drawing on techniques such as variational inference, stochastic gradient MCMC, or distributed computation.

\appendix

\section{Full Posterior Conditionals}\label{app:posteriors}

This section provides details of posterior computation for all parameters in the Bayesian hierarchical model. For brevity, the conditionals below omit the covariate terms $x_i\beta_{x,G}$ and $x_i\beta_{x,S}$; when covariates are present, each residual is adjusted accordingly.
\[
\begin{aligned}
a_{i,v,v'} &= \mu_i + \eta_{i,v}^\top \Delta \eta_{i,v'} + \epsilon_{i,v,v'}, 
\quad \epsilon_{i,v,v'} \sim \mathcal{N}(0, \sigma^2), \\
y_{i,v,k} &= \theta_{i,k} + \eta_{i,v}^\top \Omega_k \alpha_{i,k} + \delta_{i,v,k}, 
\quad \delta_{i,v,k} \sim \mathcal{N}(0, \phi_k^2), \\
\Delta &= \mathrm{diag}(\lambda_1, \dots, \lambda_R), 
\quad \Omega_k = \mathrm{diag}(\omega_{k,1}, \dots, \omega_{k,R}).
\end{aligned}
\]

We use the following parametrization for the updates of parameters \( \gamma_r \) and \( \psi_{k,r} \): we rewrite the model as \( a = w + K \lambda_r + \epsilon \) (for \( \gamma_r \)) or \( y = w + K \omega_{k,r} + \delta \) (for \( \psi_{k,r} \)), where \( a \in \mathbb{R}^{N m} \), \( y \in \mathbb{R}^{N V} \), with \( m = \frac{V(V-1)}{2} \). Vectors \( a \) and \( y \) are formed by stacking all subjects \( i = 1,\dots,N \), and all pairs \( v < v' \) (for \( a \)), or all voxels \( v \) (for \( y \)). Corresponding design matrix \( K \) stacks the terms involving \( \eta_{i,v,r}, \eta_{i,v',r} \) (for \( \lambda_r \)) or \( \eta_{i,v,r}, \alpha_{i,k,r} \) (for \( \omega_{k,r} \)), while \( w \) collects contributions from all remaining parameters (e.g., \( \mu_i, \theta_{i,k}, \lambda_{r'}, \omega_{k,r'} \), with \( r' \neq r \)).

\bigskip

Priors are specified as:

\[
\begin{aligned}
\eta_{i,v} &\sim \mathcal{N}(0, \sigma_\eta^2 I), 
\quad \alpha_{i,k} \sim \mathcal{N}(0, \sigma_\alpha^2 I), \\
\mu_i &\sim \mathcal{N}(0, 1), 
\quad \theta_{i,k} \sim \mathcal{N}(0, 1), \\
\sigma^2, \phi_k^2, \sigma_\eta^2, \sigma_\alpha^2 &\sim \mathrm{Inverse\text{-}Gamma}(a, b), \\
\lambda_r \mid \gamma_r &\sim \mathcal{N}(0, \tau_r^{-1}), 
\quad \gamma_r \sim \mathrm{Bern}(\pi_\lambda), 
\quad \tau_r \sim \mathrm{Gamma}(q^{3(r-1)}, q^{2(r-1)}), \\
\omega_{k,r} \mid \psi_{k,r} &\sim \mathcal{N}(0, \nu_{k,r}^{-1}), 
\quad \psi_{k,r} \sim \mathrm{Bern}(\pi_{\psi,k}), 
\quad \nu_{k,r} \sim \mathrm{Gamma}(q^{3(r-1)}, q^{2(r-1)}), \\
\pi_\lambda &\sim \mathrm{Beta}(a_\lambda, b_\lambda), 
\quad \pi_{\psi,k} \sim \mathrm{Beta}(a_\psi, b_\psi).
\end{aligned}
\]

\bigskip

The hierarchical model specified above leads to straightforward Gibbs sampling with full conditionals obtained as following:

\begin{itemize}

\item[•] \( \eta_{i,v} \mid - \sim \mathcal{N}(\mu_{i,v}, \Sigma_{i,v}) \) \\[0.5em]
\( \Sigma_{i,v} = \left( \frac{1}{\sigma^2} \sum_{v' \neq v} \Delta \eta_{i,v'} (\Delta \eta_{i,v'})^\top 
+ \sum_k \frac{1}{\phi_k^2} \Omega_k \alpha_{i,k} (\Omega_k \alpha_{i,k})^\top 
+ \frac{1}{\sigma_\eta^2} I \right)^{-1} \) \\[0.5em]
\( \mu_{i,v} = \Sigma_{i,v} \left( \frac{1}{\sigma^2} \sum_{v' \neq v} (a_{i,v,v'} - \mu_i) \Delta \eta_{i,v'} 
+ \sum_k \frac{1}{\phi_k^2} (y_{i,v,k} - \theta_{i,k}) \Omega_k \alpha_{i,k} \right) \)

\item[•] \( \alpha_{i,k} \mid - \sim \mathcal{N}(\mu_{\alpha_{i,k}}, \Sigma_{\alpha_{i,k}}) \) \\[0.5em]
\( \Sigma_{\alpha_{i,k}} = \left( \frac{1}{\phi_k^2} \sum_v \Omega_k^\top \eta_{i,v} \eta_{i,v}^\top \Omega_k 
+ \frac{1}{\sigma_\alpha^2} I \right)^{-1} \) \\[0.5em]
\( \mu_{\alpha_{i,k}} = \Sigma_{\alpha_{i,k}} \left( \frac{1}{\phi_k^2} \sum_v \Omega_k^\top \eta_{i,v} (y_{i,v,k} - \theta_{i,k}) \right) \)

\item[•] \( \mu_i \mid - \sim \mathcal{N} \left( \frac{B}{A}, \frac{1}{A} \right) \), \quad
\( A = \frac{V(V-1)/2}{\sigma^2} + 1 \), \quad
\( B = \frac{1}{\sigma^2} \sum_{v < v'} (a_{i,v,v'} - \eta_{i,v}^\top \Delta \eta_{i,v'}) \)

\item[•] \( \theta_{i,k} \mid - \sim \mathcal{N} \left( \frac{B}{A}, \frac{1}{A} \right) \), \quad
\( A = \frac{V}{\phi_k^2} + 1 \), \quad
\( B = \frac{1}{\phi_k^2} \sum_v (y_{i,v,k} - \eta_{i,v}^\top \Omega_k \alpha_{i,k}) \)

\item[•] \( \lambda_r \mid - \sim 
\begin{cases}
\mathcal{N} \left( \frac{B_r}{A_r}, \frac{1}{A_r} \right), & \gamma_r = 1 \\
0, & \gamma_r = 0
\end{cases} \) \\[0.5em]
\( A_r = \frac{1}{\sigma^2} \sum_i \sum_{v < v'} (\eta_{i,v,r} \eta_{i,v',r})^2 + \tau_r \) \\[0.5em]
\( B_r = \frac{1}{\sigma^2} \sum_i \sum_{v < v'} \eta_{i,v,r} \eta_{i,v',r} \left( a_{i,v,v'} - \mu_i - \sum_{r' \ne r} \lambda_{r'} \eta_{i,v,r'} \eta_{i,v',r'} \right) \)

\item[•] \( P(\gamma_r = 1 \mid -) = 
\frac{
\pi_\lambda \cdot \mathcal{N} \left( a \mid w, \sigma^2 I + \tau_r^{-1} K K^\top \right)
}{
\pi_\lambda \cdot \mathcal{N} \left( a \mid w, \sigma^2 I + \tau_r^{-1} K K^\top \right) 
+ (1-\pi_\lambda) \cdot \mathcal{N}(a \mid w, \sigma^2 I)
} \)

\item[•] \( \tau_r \mid - \sim 
\begin{cases}
\mathrm{Gamma} \left( \frac{1}{2} + q^{3(r-1)}, \ q^{2(r-1)} + \frac{1}{2} \lambda_r^2 \right), & \gamma_r = 1 \\
\mathrm{Gamma} \left( q^{3(r-1)}, \ q^{2(r-1)} \right), & \gamma_r = 0
\end{cases} \)

\item[•] \( \pi_\lambda \mid - \sim \mathrm{Beta} \left( a_\lambda + \sum_r \gamma_r, \ b_\lambda + R - \sum_r \gamma_r \right) \)

\item[•] \( \omega_{k,r} \mid - \sim 
\begin{cases}
\mathcal{N} \left( \frac{B_{k,r}}{A_{k,r}}, \frac{1}{A_{k,r}} \right), & \psi_{k,r} = 1 \\
0, & \psi_{k,r} = 0
\end{cases} \) \\[0.5em]
\( A_{k,r} = \frac{1}{\phi_k^2} \sum_i \sum_v (\eta_{i,v,r} \alpha_{i,k,r})^2 + \nu_{k,r} \) \\[0.5em]
\( B_{k,r} = \frac{1}{\phi_k^2} \sum_i \sum_v (y_{i,v,k} - \theta_{i,k} - \sum_{r' \ne r} \omega_{k,r'} \eta_{i,v,r'} \alpha_{i,k,r'}) \cdot \eta_{i,v,r} \alpha_{i,k,r} \)

\item[•] \( P(\psi_{k,r} = 1 \mid -) = 
\frac{
\pi_{\psi,k} \cdot \mathcal{N} \left( y \mid w, \phi_k^2 I + \nu_{k,r}^{-1} K K^\top \right)
}{
\pi_{\psi,k} \cdot \mathcal{N} \left( y \mid w, \phi_k^2 I + \nu_{k,r}^{-1} K K^\top \right) 
+ (1-\pi_{\psi,k}) \cdot \mathcal{N}(y \mid w, \phi_k^2 I)
} \)

\item[•] \( \nu_{k,r} \mid - \sim 
\begin{cases}
\mathrm{Gamma} \left( \frac{1}{2} + q^{3(r-1)}, \ q^{2(r-1)} + \frac{1}{2} \omega_{k,r}^2 \right), & \psi_{k,r} = 1 \\
\mathrm{Gamma} \left( q^{3(r-1)}, \ q^{2(r-1)} \right), & \psi_{k,r} = 0
\end{cases} \)

\item[•] \( \pi_{\psi,k} \mid - \sim \mathrm{Beta} \left( a_\psi + \sum_r \psi_{k,r}, \ b_\psi + R - \sum_r \psi_{k,r} \right) \)

\item[•] \( \sigma^2 \mid - \sim \mathrm{Inverse\text{-}Gamma} \left( a + \frac{N V(V-1)}{4}, \ b + \frac{1}{2} \sum_i \sum_{v < v'} (a_{i,v,v'} - \mu_i - \eta_{i,v}^\top \Delta \eta_{i,v'})^2 \right) \)

\item[•] \( \phi_k^2 \mid - \sim \mathrm{Inverse\text{-}Gamma} \left( a + \frac{N V}{2}, \ b + \frac{1}{2} \sum_i \sum_v (y_{i,v,k} - \theta_{i,k} - \eta_{i,v}^\top \Omega_k \alpha_{i,k})^2 \right) \)

\item[•] \( \sigma_\eta^2 \mid - \sim \mathrm{Inverse\text{-}Gamma} \left( a + \frac{N V R}{2}, \ b + \frac{1}{2} \sum_i \sum_v \eta_{i,v}^\top \eta_{i,v} \right) \)

\item[•] \( \sigma_\alpha^2 \mid - \sim \mathrm{Inverse\text{-}Gamma} \left( a + \frac{N P R}{2}, \ b + \frac{1}{2} \sum_i \sum_k \alpha_{i,k}^\top \alpha_{i,k} \right) \)

\end{itemize}


\bibliography{bibliography.bib}

\end{document}